\documentclass[12pt, reqno]{amsart}
\date{\today}
\usepackage{mathrsfs}               
\usepackage{amssymb}
\usepackage{graphicx, amsmath}
\usepackage{amsfonts}
\usepackage{amssymb}
\usepackage{fancyhdr}
\usepackage{a4wide}
\usepackage{xcolor}

\newcommand{\C}{\mathbb{C}} 
\newcommand{\Z}{\mathbb{Z}} 
\newcommand{\R}{\mathbb{R}} 

\newcommand{\tr}{\mathrm{Tr}}

\newtheorem{proposition}{Proposition}[section]

\newtheorem{theorem}[proposition]{Theorem}
\newtheorem{lemma}[proposition]{Lemma}

\newcommand{\sps}[2]{\langle #1,#2 \rangle} 

\newcommand{\bx}{\mathbf{x}}
\newcommand{\id}{\mathbf{1}}

\newcommand{\norm}[1]{\left\lVert #1 \right\rVert}

\newcommand{\bsigma}{\boldsymbol{\sigma}}
\newcommand{\bgamma}{\boldsymbol{\gamma}}
\newcommand{\bk}{\mathbf{k}}
\newcommand{\bm}{\mathbf{m}}

\newcommand{\F}{\mathcal{U}}
\newcommand{\fesh}{\mathcal{F}}
\newcommand{\B}{\mathcal{B}}

\newcommand{\cb}{}
\renewcommand{\le}{\leqslant}
\renewcommand{\leq}{\leqslant}
 \renewcommand{\ge}{\geqslant}
 \renewcommand{\geq}{\geqslant}     

\title[Spectral gaps in graphene antidot lattices]{Spectral gaps in graphene antidot lattices}

\author{Jean-Marie Barbaroux}

\address{Jean-Marie Barbaroux\\
 Aix Marseille Univ, Universit\'e de Toulon\\ CNRS, CPT, Marseille, France.}
\email{barbarou@univ-tln.fr}
\author{Horia Cornean}
\address{Horia Cornean\\
Department of Mathematical Sciences, Aalborg University\\
Fredrik Bajers Vej 7G, 9220 Aalborg \O, Denmark.
}
\email{cornean@math.aau.dk}
\author{Edgardo Stockmeyer}
\address{ Edgardo Stockmeyer\\ Instituto de F\'\i sica\\
Pontificia Universidad Cat\'olica de Chile\\
Vicu\~na Mackenna 4860\\
 Santiago 7820436, Chile.}
\email{stock@fis.puc.cl}

\begin{document}
\begin{abstract}
  We consider the gap creation problem in an antidot graphene lattice,
{\cb  i.e. a sheet of graphene with periodically distributed obstacles.} 
We prove several spectral results concerning {\cb the size of the gap} and its
  dependence on different natural parameters related to the antidot
  lattice.
\end{abstract}

\maketitle
\section{Introduction}

{\cb Graphene, a two-dimensional material made of carbon atoms
  arranged in a honeycomb structure, has risen a lot of attention due
  to its many unique properties. Remarkably, charge carriers close to
  the Fermi energy behave as massless Dirac fermions. This is due to
  its energy band structure which exhibits two bands crossing at the
  Fermi level making graphene a gapless semimetal
  \cite{castro2009electronic}.  Many efforts have been carried out for
  the possibility of tuning an energy gap in graphene \cite{dvorak2013bandgap}. 




 }

The main physical motivation of our work is related to the so-called
antidot graphene {\cb lattice \cite{pedersen2008graphene}}, which
consists of a regular sheet of graphene having a periodic array of
{\cb obstacles} well separated from each other. {\cb These obstacles
  can be thought, for instance, as actual holes in the graphene layer
  \cite{pedersen2008graphene}.  More generally, substrate induced
  obstacles \cite{science} or those
 created by doping or by mechanical defects have also been
  considered in the literature (see e.g.,\cite{dvorak2013bandgap} and
  references therein)}. It has been observed both experimentally and
numerically that such an array causes a band gap to open up around the
Fermi level, turning graphene from a semimetal into a 
{\cb gapped} semiconductor
{\cb (see e.g.,\cite{BPP}, \cite{science}, and \cite{dvorak2013bandgap})}.

In \cite{FPFMBPJ} {\cb and \cite{brun2014electronic} } there are given
several proposals concerning the modelling of this phenomenon. In one
of these proposals the authors replace the usual tight-binding lattice
model by a two-dimensional massless Dirac operator, while a hole is
modeled with the help of a periodic mass term.  {\cb For a large mass
  term held fixed} the authors numerically analyse how the gap
appearing near the zero energy depends on {\cb the natural parameters
  of the antidot lattice, namely,} the {\cb area occupied} by one
mass-insertion versus the area of the super unit-cell which contains
only one such hole. {\cb For holes with armchair type of boundaries,
  this model is in very good agreement with tight-binding and density
  functional ab-initio calculations \cite{brun2014electronic} (see
  also \cite{BPP, PGMP, PP})}. Moreover, the Dirac operator with a
mass term varying in a superlattice has also been used to explain the
gap appearing when  a layer of graphene is placed on substrate of
hexagonal boron nitride \cite{science} (see also
\cite{song2013electron} for the inclusion of electronic interaction).

{\cb In this article we consider the Dirac model with a periodic mass
  term and we estimate the size of the energy gap in terms of the
  strength and shape of the mass-insertion together with the natural
  parameters of the antidot lattice}. Let us now
formulate {\cb more precisely } the problem we want to investigate.

\subsection{Setting and main results}$\ $
{\cb Let 
 $\chi:\R^2\to\R$ be a bounded function supported on a compact subset $S$ included in 
$\Omega:=(-1/2,1/2]^2$ satisfying
\begin{align}
  \label{eq:chi}
\int_\Omega \chi(\bx) d\bx:=\Phi\geq 0\quad \mbox{and}\quad
\int_\Omega |\chi(\bx)|^2 d\bx=1.
\end{align}
We use the standard notation for the Pauli matrices 
\begin{align*}
\sigma_1
=\left(
\begin{array}{cc}
 0&1\\
 1&0
\end{array}
 \right),\quad
\sigma_2=\left(
\begin{array}{cc}
0&-i\\
i&0
\end{array}
\right),\quad 
\sigma_3
=\left(
\begin{array}{cc}
 1&0\\
 0&-1
\end{array}
 \right).
\end{align*}
We define the  massless free  Dirac 
operator in $L^2(\R^2,\C^2)$ as
\begin{equation}\label{hc3}
  H_0:=\frac{1}{i} {\boldsymbol \sigma}\cdot \nabla=\left ( \begin{array}{ll}
                                                              \qquad 0 & -i\partial_1-\partial_2\\
                                                              -i\partial_1+\partial_2 & \qquad  0 \end{array} \right ).
        \end{equation}
Let $\alpha\in(0,1]$ be a dimensionless parameter and let $L>0$ have
dimensions of length. In physical units,
the operator describing a mass-periodically perturbed graphene sheet is
given by 
\begin{align}\label{Htilde}
  \widetilde{H}(\alpha,\mu,L)=\hbar v_{\rm F} H_0+
\sigma_3\mu \sum_{\gamma\in \mathbb{Z}^2} \chi\Big(\frac{{\bf x}
  -\gamma L}{\alpha L}\Big),
\end{align}
where $\hbar$ is Plank's constant divided by $2\pi$ and $v_{\rm F}$ is
the Fermi velocity in graphene. Here $\mu\ge 0$ has dimensions of
energy and represents the
strength of the mass-insertion which is
$L\mathbb{Z}^2$-periodic. We note that in order for the continuum Dirac model
to hold one needs $L$ to be much larger than the distance between
the carbon atoms constituting graphene.  

As it is well known, the spectrum of $H_0$ covers the whole real
line. Our main interest in this paper is finding sufficient
conditions that the function $\chi$ must satisfy in order to create a
gap around zero in the spectrum of $\widetilde{H}$ and to estimate its
size in terms of $\alpha, \mu$, and $L$. By making the scaling transformation $\bx\mapsto L\bx$ one gets that
\begin{align}\label{relation}
   \widetilde{H}(\alpha,\mu,L)=\frac{\hbar v_{\rm F}}{L} H\Big(\alpha,
  \frac{\mu L}{\hbar v_{\rm F}}\Big).
\end{align}
Note that $ \frac{\mu L}{\hbar v_{\rm F}}$ is a dimensionless parameter.
Here, for $\beta>0$, we define the operator in  $L^2(\R^2,\C^2)$
\begin{align}\label{mop}
 H(\alpha, \, \beta) := H_0 +  \sigma_3\beta\sum_{\gamma\in\Z^2} \chi\left(\frac{\bx-\bgamma}{\alpha}\right).
\end{align}
This new operator is clearly periodic with respect to $\mathbb{Z}^2$
and (just like $\widetilde{H}$) it is self-adjoint on the first
Sobolev space $H^1(\R^2,\C^2)$ (see \cite{T}).

 Given a self-adjoint operator $T$, we denote by
$\rho(T)$ its resolvent set. Here is the first main result of our
paper.}
\begin{theorem}\label{thm:main}
Assume $\Phi\neq 0$. Then there exist two constants $C>0$ and $\delta\in (0,1)$ such that for all 
$\alpha\in(0,1/2]$ and $\beta>0$ obeying $\alpha\beta<\delta$ we have
 $$
 \left [-\alpha^2\beta(\Phi-C\alpha\beta),\,
 \alpha^2\beta(\Phi-C\alpha\beta) \right ]\subset\rho(H(\alpha,\beta)).
 $$
\end{theorem}

\vspace{0.5cm}

{\cb \noindent {\bf Remark}.  Let us comment on the consequences of
  this result regarding the energy gap, $E_{\rm g}$, for the family
  $\widetilde{H}$ defined in \eqref{Htilde}.  Let us define the area
  of the supercell $A_t:=L^2$ and $A_r:=\alpha^2 L^2$ representing the
  area supporting one mass perturbation.  In view of \eqref{relation},
  Theorem \ref{thm:main} states that for
  $\alpha \frac{\mu L}{\hbar v_{\rm F}} $ small enough and for $\Phi$
  in \eqref{eq:chi} positive
\begin{align}
    \label{eq:3.r}
    E_{\rm g}\gtrsim\, \mu \Phi\, \alpha^2 = \mu \Phi\, \frac{A_r}{A_t}.
  \end{align}
Remarkably this estimate does not depend on the side $L$ of the
supercell. This is to be contrasted with the regime of $\mu\to\infty$
considered in \cite{pedersen2008graphene} and \cite{brun2014electronic} where
it was found that, for $\alpha$ small enough (see e.g., Equation A.8
of \cite{brun2014electronic}), 
\begin{align}
  \label{eq:4.r}
  E_{\rm g}\sim \frac{\hbar v_{\rm F}}{L} \sqrt{\frac{A_r}{A_t}}.
\end{align}
}
This latter regime can be mathematically  investigated using the Dirac operator with
infinite mass boundary conditions proposed in \cite{Berry} (see
\cite{Benguria2017} for its rigorous definition).


 \vspace{0.5cm}

In the case $\Phi = 0$ it is still possible to prove the existence of a gap opening at zero. The next result needs some assumptions on $\chi$ in terms of its Fourier coefficients
\begin{equation}\label{hc1}
 \hat\chi(\bm) = \int_\Omega e^{-2 i \pi \bx\cdot \bm} \chi(\bx) d \bx , \quad\bm\in\Z^2.
\end{equation}
Note that $\Phi=0$ means that $\hat\chi(0,0)=0$. {\cb In terms of the
operator $H$ in \eqref{mop}, this time we keep $\alpha=1$ and we make
$\beta $ smaller than certain constant times the $L^\infty$-norm of
$\chi$.  Then the gap can still survive but it scales with $\beta^3$
instead of $\beta$. As a consequence, the gap for $\widetilde{H}$ has
the following behaviour, for $\frac{\mu L}{\hbar v_{\rm F}}$ small
enough,
\begin{align}
  \label{eq:3.j}
  E_{\rm g}\gtrsim \,\mu \left(\frac{\mu L}{\hbar v_{\rm F}}\right)^2.
\end{align}
}
\begin{theorem}\label{thm:main2}
Assume $\Phi=\hat\chi(0,0)=0$, and at the same time:
\begin{equation}\label{hyp-1}
 \sum_{\bm\neq 0}\sum_{\bm'\neq 0} \frac{\bm\cdot\bm'}{|\bm|^2|\bm'|^2} \, \overline{\hat\chi(\bm)}\hat\chi(\bm') 
 \hat\chi(\bm-\bm') \neq 0 .
\end{equation}
Then there exist two positive numerical constants $\beta_0$ and $C$ such that
for every $0<\beta<\beta_0/{\cb \|\chi\|_\infty}$ we have
$$
 \left [- C \beta^3 , \, C \beta^3 \right ] \subset \rho(H(1, \beta)).
$$
\end{theorem}

\vspace{0.5cm}

Let us describe a particular class of potentials where assumption \eqref{hyp-1} holds true. Assume that $\chi(\bx)$ is of the form
$$\chi(\bx)=\sum_{N-10\leq |\bm|\leq N+10} 2\cos(2\pi \bm \cdot \bx),\quad N\gg 1.$$
By construction, all the Fourier coefficients
$\hat{\chi}(\bm)=\hat{\chi}(-\bm)$ equal either $1$ or $0$. The
non-zero coefficients are those for which $\bm$ lies in an annulus
with outer radius $N+10$ and inner radius $N-10$. When $N$ becomes
large enough, the triples of vectors $\bm$, $\bm'$ and $\bm-\bm'$ for
which the Fourier coefficients in \eqref{hyp-1} are simultaneously
non-zero form a triangle which ``almost'' coincides with an
equilateral triangle with side-length equal to $N$. Here ``almost''
means that the angle between $\bm$ and $\bm'$ is close to $\pi/3$ when
$N$ is large enough. Thus the scalar product $\bm\cdot\bm'$ is
positive whenever the Fourier coefficients are non-zero (provided $N$
is large enough) and the double sum in \eqref{hyp-1} is also
positive. 

In the rest of the paper we give the proof of the two theorems listed above.


%


\section{Proof of Theorem~\ref{thm:main}}

Throughout this work we use the notation
$$ \|f\|_p=\Big(\int_\Omega|f(\bx)|^p d\bx \Big)^{1/p},\quad p\in[1,\infty).  $$ 
Note that our conditions on $\chi$ imply: 
\begin{align}\label{chi-norm}
 0< \Phi\le\|\chi\|_1\le \|\chi\|_2=1.
\end{align}

\subsection{Bloch-Floquet representation and proof of Theorem~\ref{thm:main}} $\ $

In this subsection we start by presenting the main strategy of the proof
of Theorem~\ref{thm:main}. It consists of a
suitable application of the Feshbach inversion formula to the Bloch-Floquet fiber of $H(\alpha,\beta)$. The main technical
ingredients are Lemmas \ref{lem2.2} and \ref{lem2.3} whose proof can
be found in the next subsection. At the end of this subsection we present
the proof of Theorem~\ref{thm:main}.

Let $ \mathscr{S}(\R^2,\C)$ denote the Schwartz space of test functions.  Consider the map 
\begin{align*}
&\tilde{\F}: \mathscr{S}(\R^2,\C) \subset L^2(\R^2,\C)\longrightarrow 
L^2(\Omega, L^2(\Omega,\C)), \\
&(\tilde{\F}\Psi)(\bx;\bk):=\sum_{\bgamma\in \Z^2}e^{2\pi i
  \bk\cdot(\bx+\bgamma)}\;\Psi(\bx+\bgamma), \quad  \bx,\bk\in\Omega.
\end{align*}
It is well known (see \cite{RS}) that $\tilde{\F}$ is an isometry in $ L^2(\R^2,\C)$
that can be extended to a unitary operator. We denote its unitary extension by the same
symbol. We define the Bloch-Floquet transform as
$\F:=\tilde{\F}\otimes \id_{\C^2}$.  Then we have that 
\begin{align*}
\F H(\alpha,\beta)\F^*=\int_{\Omega}^\oplus h_\bk(\alpha,\beta)
d\bk,\quad h_\bk(\alpha,\beta):= (-i\nabla_\bx^{\rm PBC}-2\pi \bk)\cdot\bsigma+\beta\chi_\alpha(\bx)\sigma_3,
  \end{align*}
  where each fiber Hamiltonian $h_\bk(\alpha,\beta)$ is defined in
  $L^2(\Omega, \C^2)$. 
Here $\nabla_\bx^{\rm PBC}$ is the gradient operator with periodic boundary
  conditions and {\cb $$\chi_\alpha(\bx):=\chi(\bx/\alpha).$$} The
  spectra of $H(\alpha,\beta)$ and $h_\bk(\alpha,\beta)$ are related through
  \begin{align}\label{spectra}
    \sigma(H(\alpha,\beta))=\overline{\bigcup_{\bk\in\Omega}\sigma(h_\bk(\alpha,\beta))}.
  \end{align}
We will use the standard  eigenbasis of $\nabla_\bx^{\rm PBC}$ given by 
\begin{align*}
  \psi_{\bm}(\bx):=e^{2\pi i \bm\cdot\bx}, \quad \bm\in \Z^2, \bx\in \Omega,
\end{align*}
which is periodic and satisfies
\begin{align*}
-i  \nabla_\bx^{\rm PBC}  \psi_{\bm}=2\pi\bm   \psi_{\bm}.
\end{align*}
For $\bm\in \Z^2$ define the projections
\begin{align}\label{eq:def-Pm}
  P_\bm:=  | \psi_{{\bm}}\rangle \langle \psi_{{\bm}}|\otimes
  \id_{\C^2}\quad \mbox{and}\quad Q_0:={\rm Id}-P_0.
\end{align}
\begin{lemma}\label{lem2.1}
Let $\alpha\in (0,1)$ and $\beta>0$. Then, for every {\cb
  $\bk\in\Omega$ and  $\psi\in P_0 L^2(\Omega,\C^2)$, we have that}
\begin{align*}
  \norm{P_0h_\bk(\alpha,\beta)P_0\psi}\ge 
\beta\alpha^2\Phi\norm{\psi}.
\end{align*}
\end{lemma}
\begin{proof}
For every $\bk \in \Omega$ we have:
\begin{align*}
  P_0h_\bk(\alpha,\beta)P_0=(-2\pi\bsigma\cdot \bk+\alpha^2\beta\Phi\sigma_3)
  P_0.
\end{align*}
Using the anticommutation relations of the Pauli matrices we get for
any $\psi\in P_0 L^2(\Omega,\C^2)$ that
\begin{align*}
  \norm{P_0h_\bk(\alpha,\beta)P_0\psi}^2=((2\pi\bk)^2+(\alpha^2\beta\Phi)^2)\norm{\psi}^2\ge (\alpha^2\beta\Phi)^2\norm{\psi}^2.
\end{align*}
\end{proof}
The previous lemma shows that $h_\bk(\alpha,\beta)$ has a spectral gap
of order $\beta\alpha^2$ on the range of $P_0$. In order to
investigate whether that is still the case on the full Hilbert space
we use the Feshbach inversion formula (see Equations (6.1) and (6.2) in \cite{N}). The latter states in this case that
$z\in \rho(h_\bk(\alpha,\beta))$ if the Feshbach operator
\begin{align*}
  \fesh_{P_0}(z):= P_0(h_\bk(\alpha,\beta) - z)P_0 -  
  \beta^2 P_0\chi_\alpha\sigma_3 Q_0 (Q_0 (h_\bk(\alpha,\beta)- z)Q_0)^{-1} Q_0 \chi_\alpha \sigma_3P_0, 
\end{align*}
is invertible on $P_0 L^2(\Omega,\C^2)$. Here we used that 
$P_0 h_\bk(\alpha,\beta) Q_0=\beta P_0 \chi_\alpha \sigma_3Q_0$.

The next lemma shows that the inverse of $Q_0 (h_\bk(\alpha,\beta)-
z)Q_0$ is well defined on the range of $Q_0$.
\begin{lemma}\label{lem2.2}
  There exists a constant $\delta\in (0,1)$ such that for all
  $\alpha\in (0,1/2)$ and $\beta>0$ with $\alpha\beta<\delta$, we have
  that $Q_0 (h_\bk(\alpha,\beta)- z)Q_0$ is invertible on the range of
  $Q_0$, for any $z\in [-\pi/2, \pi/2]$ and $\bk\in\Omega$.
\end{lemma}

\vspace{0.5cm}

The following lemma controls the second term of the Feshbach operator $\fesh_{P_0}(z)$
\begin{align*}
  \B_{P_0}(z):=\beta^2 P_0\chi_\alpha\sigma_3 Q_0 (Q_0 (h_\bk(\alpha,\beta)- z)Q_0)^{-1} Q_0 \chi_\alpha\sigma_3 P_0.
\end{align*}

\begin{lemma}\label{lem2.3}
  There exist two constants $\delta\in (0,1)$ and $ C>0$ such that for all
  $\alpha\in (0,1/2)$ and $\beta>0$ with $\alpha\beta<\delta$ we have
\begin{align*}
  \norm{\B_{P_0}(z)\psi}\le C\beta^2 \alpha^3\norm{\psi},
\end{align*}
for any $z\in [-\pi/2, \pi/2]$, $\bk\in\Omega$, and $\psi\in P_0 L^2(\Omega,\C^2)$.
\end{lemma}

\vspace{0.5cm}

Having stated all the above ingredients we can proceed to the proof of our first main
result.
\begin{proof}[Proof of Theorem \ref{thm:main}]
  In view of \eqref{spectra} it is enough to show the invertibility of
  the Feshbach operator uniformly in $\bk\in\Omega$. Using Lemmas
  \ref{lem2.1} and \ref{lem2.3} we get that for any
  $\psi\in P_0 L^2(\Omega,\C^2)$
\begin{align*}
  \norm{\fesh_{P_0}(z)\psi}&\ge  \norm{ (P_0(h_\bk(\alpha,\beta) - z)P_0 ) \psi} -\norm {\B _{P_0}(z)\psi}\\
         &\ge (\beta\alpha^2\Phi-|z|-c\alpha^3 \beta^2)\norm{\psi}.
\end{align*}
This concludes the proof by picking $\alpha\beta$ so small that
$\Phi> C\alpha\beta$.
\end{proof}
\subsection{Analysis of the Feshbach operator}$\ $

In this section we provide the proofs of Lemmas \ref{lem2.2} and
\ref{lem2.3} from the previous section. 
For that sake, let us first state some preliminary estimates.
Let
 \begin{align}
 h_\bk^{(0)}:=  (-i\nabla_\bx^{\rm PBC}-2\pi \bk)\cdot\bsigma 
 \end{align}
\begin{lemma}\label{lem3.1}
There exists a constant $C>0$, independent of
$\alpha\in (0,1)$, such that for all $\bk\in\Omega$
\begin{eqnarray}
&&\|\sqrt{|\chi_\alpha|} P_0\| \le
\alpha\label{eq:int-1}\\
&&\norm{|\chi_\alpha|^{1/2} (h_\bk^{(0)}\pm i)^{-1}
    |\chi_\alpha|^{1/2}}\le C (\alpha+\frac{\alpha^2}{1-\alpha}) \label{eq:res-est-1}\\
&&\displaystyle \norm{|\chi_\alpha|^{1/2} (h_\bk^{(0)}\pm i)^{-1}
    }\le C(\sqrt{\alpha}+\frac{\alpha}{1-\alpha}). \label{eq:res-est-2}
\end{eqnarray}
\end{lemma}
\begin{proof}
  In order to show \eqref{eq:int-1} we compute for
  $f,g\in L^2(\Omega,\C^2)$ (see also \eqref{chi-norm}):
\begin{align*}
|\sps{f}{\sqrt{|\chi_\alpha|}P_0g}| 
&\le |\sps{f}{\sqrt{|\chi_\alpha|} \psi_0}  |\,|\sps{\psi_0}{g}|
\le \|\chi_\alpha\|_1^{1/2} \|f\|_2\|g\|_2\le \alpha \|f\|_2\|g\|_2.
\end{align*}
We now turn to the proof of equations \eqref{eq:res-est-1} and 
\eqref{eq:res-est-2}.
  Denote the integral kernels of $(H_0 - i)^{-1}$ and
  $(h_\bk^{(0)}\pm i)^{-1}$ by $(H_0-i)^{-1}(\bx,\bx')$ and
  $(h_\bk^{(0)}\pm i)^{-1}(\bx,\bx')$. In the proof we will use the
  identity for $\bx\not=\bx'$
  \begin{align}\label{ec-R-r}
 (h_\bk^{(0)} -i)^{-1} (\bx,\bx') = \sum_{\gamma\in\Z^2} e^{2\pi i
    \bk\cdot (\bx+\gamma - \bx')} (H_0-i)^{-1} (\bx+\gamma,\bx'),\quad
    \bx,\bx'\in \Omega.
  \end{align}
 Let us first estimate the quadratic form of $ (h_\bk^{(0)} -i)^{-1}$
 for any $\phi,\psi\in L^2(\Omega,\C^2)$. According to
  Lemma \ref{resolvent-bound} we have
\begin{align*}
   |\sps{\phi}{(h_\bk^{(0)} -i)^{-1}\psi}|&\le C\sum_{\gamma\in\Z^2} 
\int_{\Omega\times \Omega} |\phi(\bx)|
\frac{e^{-|\bx-\bx'-\gamma|}}{|\bx-\bx'-\gamma|}
|\psi(\bx')|d\bx d\bx'\\
&\le C e^{\sqrt{2}} \sum_{\gamma\in\Z^2} e^{-|\gamma|}
\int_{\Omega\times \Omega} 
\frac{|\phi(\bx)||\psi(\bx')| }{|\bx-\bx'-\gamma|}
d\bx d\bx',
\end{align*}
where in the last step we bound the exponential using that 
$|\bx-\bx'-\gamma|\ge |\gamma|-|\bx-\bx'|\ge  |\gamma|-\sqrt{2}.
$

Now assume that the support of $\phi$ lies in
$\Omega_\alpha:=(-\alpha/2,\alpha/2]^2$, i.e., $\bx\in \Omega_\alpha$
above.  Then it is easy to check that if $\gamma\not= 0$ then
$|\bx-\bx'+\gamma|\ge (1-\alpha)$. Therefore, we find for such a case that
\begin{align*}
  |\sps{\phi}{(h_\bk^{(0)} -i)^{-1}\psi}|\le C e^{\sqrt{2}} \Big(
  \int_{\Omega\times \Omega}
\frac{ |\phi(\bx)||\psi(\bx')|}{|\bx-\bx'|}
  d\bx d\bx'
  +\frac{1}{1-\alpha}(\sum_{\gamma\not=0} e^{-|\gamma|})
  \|\phi\|_1 \|\psi\|_1   \Big).
\end{align*}
Using the Hardy-Littlewood-Sobolev inequality for $p=r=4/3$ (see Lieb and Loss Theorem
4.3) we get
\begin{equation}
  \label{eq:4.23}
  \begin{split}
 \int_{\Omega\times \Omega} 
\frac{ |\phi(\bx)| |\psi(\bx')|}
{|\bx-\bx'|}
  d\bx d\bx'\le C \|\phi\|_{4/3} \|\psi\|_{4/3} . 
  \end{split}
\end{equation}
Thus, denoting  the universal constants by $C$ we obtain, for any
$\phi,\psi\in L^2(\Omega,\C^2)$ with
${\rm supp}(\phi)\subset \Omega_\alpha$,
\begin{align}
  \label{key-bound}
   |\sps{\phi}{(h_\bk^{(0)} -i)^{-1}\psi}|\le C \big( \|\phi\|_{4/3}
  \|\psi\|_{4/3}
+(1-\alpha)^{-1}\|\phi\|_1 \|\psi\|_1 \big).
\end{align}
For some $f\in L^2(\Omega,\C^2)$
we observe that (see Remark \ref{chi-norm}) 
\begin{align}
  \label{x1}
  \||\chi_\alpha|^{1/2}f\|_1\le \| |\chi_\alpha|^{1/2}\|_2\|f\|_2\le
  \alpha  \|f\|_2.
\end{align}
Moreover, H\"older's inequality yields 
\begin{equation}
  \label{eq:3k}
\begin{split}  
  \||\chi_\alpha|^{1/2}f\|_{4/3}&\le \big(\||\chi_\alpha|^{2/3}\|_3\,
\||f|^{4/3} \|_{3/2}\big)^{3/4}\\
&\,\,=\|\chi_\alpha\|_2^{1/2} \|f\|_2\le  
\sqrt{\alpha} \|f\|_2.
\end{split}
\end{equation}
In order to get the desired bounds we recall that the norm of an
operator $T$ is given by
$\|T\|={\rm sup}_{f,g\not=0} | \sps{f}{Tg}|/(\|f\|\|g\|)$.  Hence we
find \eqref{eq:res-est-1} by using \eqref{key-bound} with
$\phi=|\chi_\alpha|^{1/2}f $ and $\psi=|\chi_\alpha|^{1/2}g$ together
with the bounds \eqref{x1} and \eqref{eq:3k}. Analogously, we obtain 
\eqref{eq:res-est-2} using again \eqref{key-bound} with
$\phi=|\chi_\alpha|^{1/2}f $ and $\psi=g$.
\end{proof}
%
%
\begin{lemma}\label{lem3.0} For any $f\in Q_0\mathcal{D}(h_\bk^{(0)})$ we have
\begin{align}
 \| h_{\bf k}^{(0)} Q_0 f\| \geq \pi\|f\| .
\end{align}
\end{lemma}
\begin{proof}
For all $\bm\in\Z^2$ and $\bk\in\Omega$, we have the identity
$$
 P_\bm h_\bk^{(0)} P_\bm = 2\pi\bsigma\cdot(\bm - \bk) P_\bm .
$$
Thus, we obtain for all $f\in Q_0\mathcal{D}(h_\bk^{(0)})$, 
$$
\| h_\bk^{(0)}Q_0 f\|^2=(2\pi)^2\sum_{\bm\not=0}\|\bsigma\cdot(\bm -
\bk) P_\bm f\|^2\ge \pi^2 \sum_{\bm\not=0}\| P_\bm f\|^2.
$$
\end{proof}
Let us introduce some notation: For a self-adjoint operator $T$ and an orthogonal projection $Q$, we define
\begin{align*}
  \rho_{Q}(T):=\left\{ z\in\C\,\mbox{ such that }\,Q(T-z)Q:
{\rm Ran}Q\to  {\rm Ran}
Q\quad \mbox{is invertible}\right\}.
\end{align*} 
We set 
\begin{align*}
  &R_0(z):=\left(Q_0(h_{\bf k}^{(0)}-z)Q_0\!\!\upharpoonright_{{\rm
  Ran}Q_0} \right)^{-1}, \quad z\in  \rho_{Q_0}(h_{\bf k}^{(0)}),\\
&R(z):=\left(Q_0(h_{\bf k}(\alpha,\beta)-z)Q_0\!\!\upharpoonright_{{\rm
  Ran}Q_0} \right)^{-1}, \quad z\in  \rho_{Q_0}(h_{\bf k}(\alpha,\beta)).
\end{align*}
Let $U:L^2(\Omega,\C^2) \to \mathrm{Ran} Q_0$ and $W: \mathrm{Ran}Q_0 \to L^2(\Omega,\C^2)$ be defined as 
\begin{align*}
 & W := \sqrt{\beta} \left(\sqrt{|\chi_\alpha}|
 \sigma_3\right) Q_0, \\
 & U: = \sqrt{\beta} Q_0 \left(\mathrm{sgn}(\chi_\alpha) \sqrt{|\chi_\alpha|} \right) .
\end{align*}
%
%
%
%
\begin{lemma}\label{lem3.2} There exists $C>0$, independent of 
$\alpha\in (0,1/2]$ and $\beta>0$,
such that for all $|z|\leq \pi/2$,
$$
 \| W R_0(z) U \| < C\alpha\beta .
$$
\end{lemma}
%
%
%
\begin{proof}
Note that due to Lemma \ref{lem3.0} we have for  $|z|\leq \pi/2$ 
\begin{align}\label{something}
\|R_0(z)\|\le \frac{2}{\pi}.
\end{align}
Using the first resolvent identity
\begin{equation}\label{eq:re-1}
\begin{split}
 W R_0(z) U = W R_0(i) U + (z-i) W R_0(i)^2 U + (z-i)^2 W R_0(i)R_0(z)R_0(i)  U .
\end{split}
\end{equation}
We shall estimate separately each term on the right hand side of \eqref{eq:re-1}.\\
Since $h_\bk^{(0)}$ commutes with the projections $P_\bm$ we have
\begin{equation}\label{eq:re-4}
 R_0(i) =    
 (h_\bk^{(0)} - i )^{-1}   
 -  (P_0  (h_\bk^{(0)} - i )P_0\!\!\upharpoonright_{{\rm Ran}P_0} )^{-1} .
\end{equation}
Thus, using \eqref{eq:int-1}, we obtain the estimate 
\begin{equation}\label{eq:re-added}
\begin{split}
\| W (P_0 (h_\bk^{(0)} -i)P_0\!\!\upharpoonright_{{\rm Ran}P_0})^{-1}U\| 
& \leq \beta \|\sqrt{|\chi_\alpha|} P_0\|\, \|(h_\bk^{(0)} -i)^{-1}\|\, 
\| \sqrt{P_0|\chi_\alpha|}\| \\ 
& \leq \beta \alpha^2.
\end{split}
\end{equation}
The identity \eqref{eq:re-4} together with the inequalities 
\eqref{eq:res-est-1} and \eqref{eq:re-added}, imply that there are
universal constants $c, C>0$, such that  for $|\alpha|<1/2$
\begin{equation}\label{eq:re-8}
\begin{split}
\| W R_0(i)U \| & \leq \| W (h_\bk^{(0)} -i)^{-1}U \| 
+ \| W (P_0 (h_\bk^{(0)} -i)P_0\!\!\upharpoonright_{{\rm Ran}P_0})^{-1}U\| \\
& \leq c  \beta\alpha  + \beta \alpha^2\leq C\beta\alpha.
\end{split}
\end{equation}
This bounds the first term on the right hand side of \eqref{eq:re-1}.
To estimate the second one we first notice that
\begin{equation*}
\begin{split}
  (h_\bk^{(0)} - i)^{-2}   =  
  (Q_0  (h_\bk^{(0)} - i )Q_0\!\!\upharpoonright_{{\rm Ran}Q_0})^{-2} 
  + (P_0  (h_\bk^{(0)} - i )P_0\!\!\upharpoonright_{{\rm Ran}P_0})^{-2}
\end{split}
\end{equation*}
Therefore, using the same strategy as in \eqref{eq:re-added}, there exists $C>0$ independent of $\alpha$ and $\beta$ such that 
\begin{equation}\label{eq:re-9}
\begin{split}
 \| (z-i) W R_0(i)^2 U \| & = \left\| (z-i) W \left( (h_\bk^{(0)} - i )^{-2} 
 - (P_0  (h_\bk^{(0)} - i )P_0\!\!\upharpoonright_{{\rm Ran}P_0})^{-2}\right) U\right\| \\
 & \leq \| (z-i) W  (h_\bk^{(0)} - i )^{-2} U\| + |z-i|\beta\alpha^2 \\
 & \leq \pi\beta \| \sqrt{|\chi_\alpha |}(h_\bk^{(0)} - i )^{-1} \| \, \|(h_\bk^{(0)} - i )^{-1} \sqrt{|\chi_\alpha|} \| + \pi\beta\alpha^2 \\
 & \leq C \beta \alpha,
\end{split}
\end{equation}
where we used \eqref{eq:res-est-2} in the last inequality.

 Finally, we bound the last term on the right hand side of
 \eqref{eq:re-1}. Observe that from Lemma~\ref{lem3.1} and inequality \eqref{eq:int-1},
we obtain that there exists $c>0$ such that for all $\alpha\le 1/2$
\begin{align*}
 \| \sqrt{|\chi_\alpha|} R_0(i) \| \leq \| \sqrt{|\chi_\alpha|} (h_\bk^{(0)}-i)^{-1}\|
 + \| \sqrt{|\chi_\alpha|} P_0 (P_0 (h_\bk - i)P_0)^{-1} \| \leq c 
\sqrt{\alpha}.
\end{align*}
Therefore, using \eqref{something} and \eqref{eq:res-est-2}
\begin{align*}
 \| (z-i)^2 W R_0(i)R_0(z) R_0(i) U \|  &\leq \beta |z-i|^2
 \| \sqrt{|\chi_\alpha|} R_0(i) \|\, \| R_0(z)\|\, \|R_0(i)
 \sqrt{|\chi_\alpha|}\|\\
 & \leq  C \beta |z-i|^2 \alpha. 
\end{align*}
In view of \eqref{eq:re-1}, the latter bound together with \eqref{eq:re-8} and
\eqref{eq:re-9}  concludes the proof.
\end{proof}
Before stating the next lemma we define the set
\begin{align*}
  \mathcal{S}:=\{z\in  \rho_{Q_0}(h_{\bf k}^{(0)})\,:\, \|WR_0(z)U\|<1\}.
\end{align*}
\begin{lemma}\label{resolventidentity}  
For any $z\in \mathcal{S}$ we have that
  $z\in \rho_{Q_0}(h_{\bf k}(\alpha,\beta))$ and 
  \begin{align*}
    R(z)=R_0(z)-R_0(z)U(1+WR_0(z)U)^{-1}W R_0(z).
  \end{align*}
\end{lemma}
\begin{proof}
For $z\in \mathcal{S} $ we define for short 
\begin{align*}
  B(z):=R_0(z)-R_0(z)U(1+WR_0(z)U)^{-1}W R_0(z).
\end{align*}
Since $\|WR_0(z)U\|<1$ we may use Neumann series to get
\begin{equation}
  \label{eq:2}
\begin{split}
  B(z) &= R_0(z)-R_0(z)U\Big( \sum_{n\ge 0 }(-1)^n(WR_0(z)U)^n \Big)WR_0(z)\\
&=  R_0(z)\Big[ 1-\sum_{n\ge 0 }(-1)^nU (WR_0(z)U)^n WR_0(z) \Big]
=R_0(z) \sum_{n\ge 0 }(-1)^n(U W R_0(z))^n, 
\end{split}
\end{equation}
where the absolute convergence of the last expression is a consequence of the identity:
 $$(U W R_0(z))^n=U(W R_0(z)U)^{n-1}W R_0(z),\quad n\geq 1,$$
 and the boundedness of $U$ and $W$. 

  For $z\in \mathcal{S}\cap\R $ we define
  $z_\varepsilon:=z+i\varepsilon$. Since $\mathcal{S}$ is an open set we may choose  
  $\varepsilon>0$  so small that $z_\varepsilon\in \mathcal{S}$. We
  first prove that the claim holds for  $z_\varepsilon$.  A simple
  iteration of the second resolvent identity gives
  \begin{align}
    \label{eq:1}
 R(z_\varepsilon)=R_0(z_\varepsilon)\sum_{n\ge 0}^N (-1)^n \big(UWR_0(z_\varepsilon)\big)^n+T_{N+1},   
  \end{align}
where 
$$T_{N+1}=(-1)^{N+1}(UWR_0(z_\varepsilon))^N UW
R(z_\varepsilon)=(-1)^{N+1} U [WR_0(z_\varepsilon)U]^{N}WR(z_\varepsilon).$$  Hence $\Vert T_{N+1}\Vert $ converges to zero and  $R(z_\varepsilon)=B(z_\varepsilon)$.
Taking the limit $\varepsilon\to 0$ on the right hand side of this identity finishes the proof.
\end{proof}
\begin{proof}[Proof of Lemmas \ref{lem2.2} and \ref{lem2.3}]
Notice that the proof of Lemma~\ref{lem2.2} follows from
Lemma~\ref{resolventidentity} since  $z\in \mathcal{S}$ provided
$\alpha\beta$ is small enough (see Lemma~\ref{lem3.2}). 

In order to show Lemma~\ref{lem2.3} observe that 
\begin{align}
  \label{eq:3.23}
  \beta^2\|P_0\chi_\alpha R_0(z)\chi_\alpha P_0 \|\le \beta\|P_0
|\chi_\alpha|^{1/2}\| \|WR_0(z) U\| \||\chi_\alpha|^{1/2} P_0\|\le c\beta^2\alpha^3,
\end{align}
where we used Lemma \ref{lem3.2} and \eqref{eq:int-1}. Moreover,
assuming $\alpha\beta$ is so small that $ \|WR_0(z) U\|<1/2$ we have
\begin{align*}
   &\beta^2\|P_0 \chi_\alpha 
R_0(z)  U(1+WR_0(z)U)^{-1}W R_0(z)
\chi_\alpha P_0\|\\
&\le \beta \|P_0|\chi_\alpha|^{1/2}\|  \|WR_0(z) U\|
\| (1+WR_0(z)U)^{-1}\|  \|WR_0(z) U\|
\||\chi_\alpha|^{1/2} P_0\|\le c\beta^3\alpha^4.
\end{align*}
The latter inequality together with \eqref{eq:3.23} finishes the proof
of Lemma~\ref{lem2.3} in view of the resolvent identity of Lemma~\ref{resolventidentity}.
\end{proof}



\section{Proof of Theorem~\ref{thm:main2}}
In this section we remind that we fix $\alpha=1$ in the definition of $h_\bk(\alpha, \beta)$. Let us redenote the Feshbach operator $\fesh_{P_0}(z)$ by $\fesh_{\bk}(z)$ to emphasize its dependence on the vector $\bk\in\Omega$:
\begin{align*}
  \fesh_{\bk}(z):= P_0(h_\bk(1,\beta) - z)P_0 -  
  \beta^2 P_0\chi\sigma_3 Q_0 (Q_0 (h_\bk(1,\beta)- z)Q_0)^{-1} Q_0 \chi \sigma_3P_0 .
\end{align*}
We shall later on prove that $\fesh_{\bk}(0)$ is invertible for all $\bk\in \Omega$, with an inverse uniformly bounded in $\bk$. We now show that this information is enough for the existence of a gap near zero for the original operator $H(1,\beta)$. 
\begin{lemma}\label{lemma3.1}
Assume that there exists a constant $C>0$ such that
\begin{equation}\label{eq:est-b}
 \sup_{\bk \in\Omega} \| \fesh_\bk(0)^{-1} \| \leq C/\beta^3 .
\end{equation}
Then there exists a constant $\tilde{C}>0$ such that for all $|z| \leq \tilde{C}\beta^3$, we have $z\in\rho\left( H(1, \beta)\right)$.
\end{lemma}
\begin{proof}
There exists $C'$ such that 
$$
 \sup_{\bk\in\Omega} \| \fesh_\bk(z) - \fesh_\bk(0)\| \leq C' |z|.
$$
Using \eqref{eq:est-b} yields
$$
 \sup_{\bk\in\Omega} \| (\fesh_\bk(z) - \fesh_\bk(0))\fesh_\bk(0)^{-1} \| \leq C' |z| 
 \frac{C}{\beta^3} .
$$
Hence, $\fesh_\bk(z) = \left[ 1 + (\fesh_\bk(z) - \fesh_\bk(0))\fesh_\bk(0)^{-1}\right]\fesh_\bk(0)$
is invertible for $|z|\leq \frac{\beta^3}{2 C C'}$ uniformly in $\bk\in\Omega$. This implies 
that $z\in\rho(H(1,\beta)$.
\end{proof}

Now we focus on proving the estimate \eqref{eq:est-b}. For that sake we consider two regimes in $\bk$.
\begin{lemma}\label{lemma3.2}
Let $\beta\in (0, \pi / (2 \|\chi\|_\infty) )$ and 
 let $\bk\in\Omega$ such that $|\bk| > 2\beta^2 / \pi^2$. Then $\fesh_{\bk}(0)$ is invertible and 
\begin{equation}\label{eq:bnd1}
 \sup_{\bk\in\Omega,\, |\bk| > 2\beta^2 / \pi^2}\| \fesh_{\bk}(0) ^{-1} \| \leq \frac{\pi}{2\beta^2}. 
\end{equation}
\end{lemma}
\begin{proof}
From Lemma~\ref{lem3.0} we have  
$\|h_\bk(1,0)  f \| \geq \pi \|f\|$ for all $f\in Q_0\mathcal{D}(h_\bk(1,0))$. 
Thus for $\beta\in (0, \pi / ( 2 \|\chi\|_\infty) )$: 
\begin{align*}
 \|Q_0 h_\bk(1,\beta)  f\| & = \| (h_\bk(1,0)) f + Q_0 \beta\chi\sigma_3  f\|
  \geq \pi \| f\| - \beta \|\chi\|_\infty \| f\| \geq \frac{\pi}{2} \|f\|.
\end{align*} 
Thus using $\|\chi\|_2 = 1$:
\begin{equation}\label{eq:fesh-1}
\begin{split}
 \| \fesh_\bk(0) + 2\pi \bsigma\cdot\bk  P_0\| 
& = \beta^2 \| \sigma_3 P_0 \chi Q_0 [Q_0 h_\bk(1,\beta) Q_0]^{-1} Q_0 \chi P_0 \sigma_3\| \\
& \leq \beta^2  \| [Q_0 h_\bk (1,\beta) Q_0]^{-1} \| \leq \frac{2}{\pi}\beta^2  .
\end{split}
\end{equation}
We have the identity
\begin{align}\label{hc2}
  (P_0 2\pi \bsigma\cdot\bk P_0)^{-1} = P_0 \frac{1}{2\pi \bk^2}\bsigma\cdot \bk ,
\end{align}
hence under our assumption on $\bk$ we get that the operator 
$1 - (\fesh_\bk(0) + 2\pi \bsigma\cdot\bk P_0)(P_0 2\pi \bsigma\cdot\bk P_0)^{-1}$ is invertible, and
\begin{align*}
  \| \fesh_\bk(0)^{-1}\| 
 = \left\| P_0 \frac{1}{2\pi \bk^2}\sigma\cdot \bk \Big( 1 - \big(\fesh_\bk(0) +2\pi \sigma\cdot\bk P_0\big)P_0 \frac{1}{2\pi \bk^2}\sigma\cdot \bk \Big)^{-1} \right\|
\leq \frac{\pi}{2\beta^2}
\end{align*}
which proves \eqref{eq:bnd1}.
\end{proof}

Now we focus on the case $|\bk| \leq 2\beta^2 /\pi^2$. Applying the resolvent formula to the operator $\fesh_\bk(0)$ yields
\begin{equation}\label{eq-est-fesh-2}
\begin{split}
 \fesh_\bk(0) 
 & =  -2\pi \bsigma\cdot\bk P_0 - \beta^2 \sigma_3 P_0 \chi Q_0 \left( Q_0 h_0(1,\beta) Q_0\right)^{-1} Q_0 \chi P_0 \sigma_3 \\
   & \ \ \ + \beta^2 \sigma_3 P_0 \chi Q_0 \left( Q_0 h_\bk(1,\beta) Q_0\right)^{-1} 
   2\pi \bsigma\cdot\bk  \left( Q_0 h_0(1,\beta) Q_0\right)^{-1}  Q_0 \chi P_0 \sigma_3 \\
   & =: M_\bk(\beta) + \mathcal{O}(\beta^2 |\bk| ) ,
\end{split}
\end{equation}
where we used  \eqref{eq:fesh-1} to prove that the third term on the right hand side of the first equality is $\mathcal{O}(\beta^2 |\bk| )$.
The operator $M_\bk(\beta)$ has the structure $P_0\otimes m_\bk(\beta)$ where its matrix part $m_\bk(\beta)$ is acting on $\C^2$.
\begin{lemma} The matrix $m_\bk(\beta)$ is traceless, i.e. 
$
\tr (m_\bk(\beta)) = 0.
$
In particular, defining the three dimensional vector
$$
 W_\bk(\beta) := \big(\tr(\sigma_1 m_\bk(\beta)),\, \tr(\sigma_2 m_\bk(\beta)),\,
 \tr(\sigma_3 m_\bk(\beta))   \big)\in\R^3,
$$
we have
\begin{equation}\label{eq:M_k-dec}
 m_\bk(\beta) = W_\bk(\beta)\cdot\bsigma.
\end{equation}
\end{lemma}
\begin{proof}
Consider the anti-unitary charge conjugation operator $U_c$ defined by 
$
 U_c \psi = \sigma_1 \overline{\psi}.
$
Then a straightforward computation yields $U_c^2=1$ and
\begin{align*}
 U_cP_0=P_0U_c,\quad U_c\sigma_3=-\sigma_3U_c,\quad U_c h_0(1,\beta)  = - h_0(1,\beta) U_c . 
\end{align*}
Also:
\begin{align*}
 & U_c \left[\sigma_3 P_0 \chi Q_0 \left( Q_0 (h_0(1,\beta)Q_0\right)^{-1} 
 Q_0 \chi P_0 \sigma_3\right] U_c  = - \sigma_3 P_0 \chi Q_0 \left( Q_0 (h_0(1,\beta)Q_0\right)^{-1} 
 Q_0 \chi P_0 \sigma_3,
\end{align*}
Since $\tr (U_c A U_c) = \tr(A)$ if $A$ is self-adjoint, we must have $\tr (m_\bk(\beta)+2\pi \bsigma\cdot\bk) = 0$. Since $\tr ( \bsigma\cdot\bk) = 0$, we conclude that $
\tr (m_\bk(\beta)) = 0
$ and the lemma is proved.
\end{proof}

\begin{lemma}
Assume that hypothesis~\eqref{hyp-1} of Theorem~\ref{thm:main2} holds. Then there exists $\beta_0 >0$ and $C>0$  such that for all $\beta\in (0, \beta_0)$, and all $\bk\in\Omega$ such that $|\bk| \leq 2 \beta^2 /\pi^2$, $M_\bk(\beta)$ is invertible and 
\begin{equation}\label{eq:decadix}
  \| M_\bk(\beta)^{-1} \| \leq C/\beta^3 . 
\end{equation}  
\begin{proof}
Due to \eqref{eq:M_k-dec}, we have
$$
 M_\bk(\beta)^{-1} = \frac{1}{|W_\bk(\beta)|^2} M_\bk(\beta),
$$
hence
\begin{equation}\label{eq:est-trace-1}
\| M_\bk(\beta)^{-1} \| = \frac{1}{|W_\bk(\beta)|} \leq \frac{1}{|\tr ( \sigma_3 M_\bk(\beta))|} .
\end{equation}
Let us now compute (remember that the "third" component of $\bk$ is by definition equal to zero): 
\begin{equation}\label{eq:est-trace-2}
\begin{split}
\tr(\sigma_3 M_k(\beta)) 
& = \tr\left(\sigma_3 \left(-2\pi \bsigma\cdot\bk P_0 - \beta^2 \sigma_3 P_0 \chi Q_0 \left( Q_0 h_0(1,\beta) Q_0\right)^{-1} Q_0 \chi P_0 \sigma_3\right) \right) \\
& = -\beta^2 \tr \left( P_0 \chi Q_0 \left( Q_0 h_0(1,\beta) Q_0\right)^{-1} Q_0 \chi P_0 \sigma_3 \right) =: -\beta^2 w(\beta),
\end{split}
\end{equation}
which is independent of $\bk$. 
Using \eqref{eq:def-Pm} we get
\begin{align*}
 \big(Q_0 h_0(1,0) Q_0\big)^{-1} 
 & = \sum_{\bm\neq 0}\frac{1}{2\pi |\bm|^2} P_\bm \bsigma\cdot\bm =\sum_{\bm\neq 0} \frac{1}{2\pi |\bm|^2}P_\bm (m_1\sigma_1 +m_2\sigma_2).
\end{align*}
Hence we obtain $w(0)=0$ because 
\begin{equation}\label{eq:est-trace-3}
 \tr \left( P_0 \chi Q_0 \left( Q_0 h_0(1,0) Q_0\right)^{-1} Q_0 \chi P_0 \sigma_3 \right) = 0. 
\end{equation}
Moreover
\begin{equation}\label{eq:est-trace-4}
\begin{split}
 w'(0) & = \tr \left(P_0 \chi Q_0 \left( Q_0 h_0(1,0) Q_0 \right)^{-1} 
 \chi\sigma_3  \left( Q_0 h_0(1,0) Q_0 \right)^{-1} Q_0 \chi P_0 \sigma_3\right) \\
 & = \sum_{\bm\neq 0}\sum_{\bm'\neq 0}
 \frac{\overline{\hat\chi(\bm)}}{2\pi |\bm|^2} \hat\chi(\bm - \bm') 
 \frac{\hat\chi(\bm')}{2\pi |\bm'|^2} 
 \tr\big( (\bm\cdot\bsigma) \sigma_3 (\bm'\cdot\bsigma)\sigma_3 \big) \\
 & = - \sum_{\bm\neq 0}\sum_{\bm'\neq 0}
 \frac{\overline{\hat\chi(\bm)}}{2\pi |\bm|^2} \hat\chi(\bm - \bm') 
 \frac{\hat\chi(\bm')}{2\pi |\bm'|^2}  \bm\cdot\bm'
\end{split}
\end{equation}
With equations \eqref{eq:est-trace-2}-\eqref{eq:est-trace-4} we get
$$
\tr(\sigma_3 M_k(\beta)) = \frac{\beta^3}{4\pi^2}
\sum_{\bm\neq 0}\sum_{\bm'\neq 0} \frac{\bm\cdot\bm'}{|\bm|^2|\bm'|^2} \, \overline{\hat\chi(\bm)}\hat\chi(\bm') 
 \hat\chi(\bm-\bm')  +\mathcal{O}(\beta^4) ,
$$
which together with \eqref{eq:est-trace-1} concludes the proof of the lemma since we assumed hypothesis~\eqref{hyp-1}. 
\end{proof}
\end{lemma}

\begin{proof}[Proof of Theorem~\ref{thm:main2}]
Equation \eqref{eq-est-fesh-2} together with \eqref{eq:decadix} implies that for $\beta\in (0, \beta_0)$ and $|\bk|\leq 2\beta^2/\pi^2$, the operator $\fesh_\bk(0)$ is invertible and 
$$
 \sup_{\bk\in\Omega,\, |\bk| \leq 2\beta^2 / \pi^2}\| \fesh_\bk(0)^{-1}  \|\leq C/\beta^3
$$ 
for some constant $C>0$ independent of $\beta$. Using in addition the estimate \eqref{eq:bnd1} of  Lemma~\ref{lemma3.2} it implies that there exists a constant $C>0$ independent of $\beta$ such that
$$
 \sup_{\bk\in\Omega} \| \fesh_\bk(0)^{-1}  \| \leq C/\beta^3.
$$
Together with Lemma~\ref{lemma3.1}, this concludes the proof of Theorem~\ref{thm:main2}.
\end{proof}


\appendix
\section{Estimate for the  resolvent kernel}
\begin{lemma}\label{resolvent-bound}
There exists a constant $C>0$ such that the following kernel estimate holds for all $\bx\neq\bx' \in\R^2$
 $$
  | (H_0\pm i)^{-1} (\bx, \bx')  | \leq C \frac{e^{-|\bx - \bx'|}}{|\bx-\bx'|}.
 $$
\end{lemma}
\begin{proof} 
  The relation
  $(H_0+i)(H_0-i) = (-\Delta + 1)$ implies that
  \begin{align}\label{ec1}
     (H_0- i)^{-1}= (H_0+i) (-\Delta + 1)^{-1}.
  \end{align}
In order to obtain the kernel for the Laplacian we recall the
well-known formula for its heat kernel in dimension $n\ge1$
\begin{align*}
  e^{t\Delta}(\bx,\bx')=\frac{1}{(4\pi t)^{n/2}} e^{-\tfrac{|\bx-\bx'|^2}{4t}}.
\end{align*}
Thus, by the usual integral representation of the resolvent in terms
of the heat kernel we get
\begin{align*}
   (-\Delta + 1)^{-1}(\bx,\bx')=\frac{1}{4\pi}\int_0^\infty
  e^{-\tfrac{|\bx-\bx'|^2}{4t} } e^{-t} \frac{dt}{t}= \frac{1}{2\pi} K_0(|\bx-\bx'|),
\end{align*}
where $K_0$ is a modified Bessel function (see Formula
8.432(6) in \cite{GR}). Using this in \eqref{ec1} and that $K_0'(z)=-K_1(z)$ (see
Formula 8.486(18) in \cite{GR}) we get that
\begin{align*}
   (H_0\pm i)^{-1}(\bx,\bx')=\frac{i}{2\pi} \big(K_1(|\bx-\bx'|){\boldsymbol \sigma}
  \cdot \nabla_{\bx} |\bx-\bx'|\mp K_0(|\bx-\bx'|)\big).
\end{align*}
Thus 
$\left\vert (H_0\pm i)^{-1}(\bx,\bx')\right\vert\le
\frac{1}{2\pi}\big(\left \vert K_0(|\bx-\bx'|)\right \vert +\left \vert K_1(|\bx-\bx'|)\right \vert\big)$
the claim now follows by the asymptotic behaviour of the Bessel
functions $K_0$ and $K_1$ at zero and infinity (see Formulas 8.447(3), 8.446 and 8.451(6) in \cite{GR}).
\end{proof}
\bigskip

\noindent
{\cb {\bf Acknowledgments.}
It is a pleasure to thank the REB program of CIRM for giving us the
opportunity to start this research.  Furthermore, we thank the
Pontificia Universidad Cat\'olica de Chile, Aalborg Universitet and
Universit\'e de Toulon for their hospitality.
E.S has been partially funded by Fondecyt (Chile)
project \# 114--1008 and  Iniciativa Cient\'ifica Milenio (Chile) through the
Nucleus RC–120002. }

\bibliographystyle{plain}

\end{document}